\documentclass[10pt,journal,twocolumn,twoside]{IEEEtran}
\usepackage{amsfonts,amsbsy,bm,amsmath}
\usepackage[nolist]{acronym}
\usepackage{mathrsfs} 
\usepackage{graphicx,cite,amssymb,mathtools,amsthm}
\usepackage{subcaption}
\usepackage{stfloats}
\usepackage{xcolor}
\mathtoolsset{showonlyrefs}
\usepackage{bbm}
\allowdisplaybreaks

\DeclareCaptionLabelSeparator{periodspace}{.\quad}
\newcommand{\ensemble}{\mathscr{C}}
\newcommand{\code}{\mathcal{C}}
\newcommand{\vecu}{\boldsymbol{u}}

\newcommand{\vecf}{\boldsymbol{f}}
\newcommand{\vecv}{\boldsymbol{v}}

\newcommand{\vecc}{\boldsymbol{c}}

\newcommand{\vecy}{\boldsymbol{y}}

\newcommand{\F}{\boldsymbol{F}}
\newcommand{\G}{\boldsymbol{G}}
\newcommand{\GK}{\boldsymbol{K}}

\newcommand{\Gnsys}{\G_{\mathsf{i},\mathsf{nsys}}}

\newtheorem{prop}{Proposition}

\newtheorem{remark}{Remark}
\newtheorem{example}{Example}
\newtheorem{definition}{Definition}

\definecolor{lightblue}{rgb}{0,.5,1}
\definecolor{normemph}{rgb}{0,.2,0.6}
\definecolor{supremph}{rgb}{0.6,.2,0.1}
\definecolor{lightpurple}{rgb}{.6,.4,1}
\definecolor{gold}{rgb}{.6,.5,0}
\definecolor{orange}{rgb}{1,0.4,0}
\definecolor{hotpink}{rgb}{1,0,0.5}
\definecolor{newcolor2}{rgb}{.5,.3,.5}
\definecolor{newcolor}{rgb}{0,.3,1}
\definecolor{newcolor3}{rgb}{1,0,.35}
\definecolor{darkgreen1}{rgb}{0, .35, 0}
\definecolor{darkgreen}{rgb}{0, .6, 0}
\definecolor{darkred}{rgb}{.75,0,0}
\definecolor{midgray}{rgb}{.8,0.8,0.8}
\definecolor{darkblue}{rgb}{0,.25,0.6}

\definecolor{lightred}{rgb}{1,0.9,0.9}
\definecolor{lightblue}{rgb}{0.9,0,0.0}
\definecolor{lightpurple}{rgb}{.6,.4,1}
\definecolor{gold}{rgb}{.6,.5,0}
\definecolor{orange}{rgb}{1,0.4,0}
\definecolor{hotpink}{rgb}{1,0,0.5}
\definecolor{darkgreen}{rgb}{0, .6, 0}
\definecolor{darkred}{rgb}{.75,0,0}
\definecolor{darkblue}{rgb}{0,0,0.6}

\definecolor{bgblue}{RGB}{245,243,253}
\definecolor{ttblue}{RGB}{91,194,224}

\definecolor{dark_red}{RGB}{150,0,0}
\definecolor{dark_green}{RGB}{0,150,0}
\definecolor{dark_blue}{RGB}{0,0,150}
\definecolor{dark_pink}{RGB}{80,120,90}

\begin{acronym}
	\acro{AWGN}{input additive white Gaussian noise}
	\acro{B-AWGN}{binary input additive white Gaussian noise}
	\acro{B-DMC}{binary-input discrete memoryless channel}
	\acro{B-DMSC}{binary-input discrete memoryless symmetric channel}
	\acro{B-MSC}{binary-input memoryless symmetric channel}
	\acro{BCJR}{Bahl, Cocke, Jelinek, and Raviv}
	\acro{BEC}{binary erasure channel}
	\acro{BER}{bit error rate}
	\acro{BLEP}{block error probability}
	\acro{BP}{belief propagation}
	\acro{BSC}{binary symmetric channel}
	\acro{CER}{codeword error rate}
	\acro{CN}{check node}
	\acro{CRC}{cyclic redundancy check}
	\acro{DE}{density evolution}
	\acro{GA}{Gaussian approximation}
	\acro{IO-WE}{input-output weight enumerator}
	\acro{IR-WE}{input-redundancy weight enumerator}
	\acro{IO-WEF}{input-output weight enumerating function}
	\acro{IR-WEF}{input-redundancy weight enumerating function}
	\acro{JIO-WE}{joint IO-WE}
	\acro{JIR-WE}{joint IR-WE}
	\acro{JWE}{joint WE}
	\acro{LDPC}{low-density parity-check}
	\acro{LHS}{left-hand side}
	\acro{LLR}{log-likelihood ratio}
	\acro{MAP}{maximum a posteriori}
	\acro{MC}{metaconverse}
	\acro{ML}{maximum likelihood}
	\acro{PC}{product code}
	\acro{pdf}{probability density function}
	\acro{RCB}{random coding bound}
	\acro{RCUB}{random coding union bound}
	\acro{RM}{Reed--Muller}
	\acro{RHS}{right-hand side}
	\acro{RV}{random variable}
	\acro{SPC}{single parity-check}
	\acro{SC}{successive cancellation}
	\acro{SCC}{super component codes}
	\acro{SCL}{successive cancellation list}
	\acro{SISO}{soft-input soft-output}
	\acro{SNR}{signal-to-noise ratio}
	\acro{UB}{union bound}
	\acro{TUB}{truncated union bound}
	\acro{VN}{variable node}
	\acro{WE}{weight enumerator}
	\acro{WEF}{weight enumerating function}
\end{acronym}

\begin{document}

	\title{Successive Cancellation List Decoding of Product Codes with Reed-Muller Component Codes}
	\author{Mustafa Cemil Co\c{s}kun, \IEEEmembership{Student Member, IEEE}, Thomas Jerkovits, \IEEEmembership{Student Member, IEEE},\\ Gianluigi Liva, \IEEEmembership{Senior Member, IEEE}
		\thanks{This work was supported by Munich Aerospace e.V., grant "Efficient Coding and Modulation for Satellite Links with Severe Delay Constraints".}
		\thanks{Mustafa Cemil Co\c{s}kun, Thomas Jerkovits and Gianluigi Liva are with the Institute of Communications and	Navigation of the German Aerospace Center (DLR), M\"unchner Strasse 20, 82234 We{\ss}ling, Germany (email: \{mustafa.coskun,thomas.jerkovits,gianluigi.liva\}@dlr.de).}
	} 
	
	\maketitle
	
	\begin{abstract}
		This letter proposes successive cancellation list (SCL) decoding of product codes with Reed--Muller (RM) component codes. SCL decoding relies on a product code description based on the $2\times 2$ Hadamard kernel, which enables interpreting the code as an RM subcode. The focus is on a class of product codes considered in wireless communication systems, based on single parity-check and extended Hamming component codes. For short product codes, it is shown that SCL decoding with a moderate list size performs as well as (and, sometimes, outperforms) belief propagation (BP) decoding. Furthermore, by concatenating a short product code with a high-rate outer code, SCL decoding outperforms BP decoding by up to $1.4$ dB.
	\end{abstract}
	
	\section{Introduction}\label{sec:intro}
	
	\IEEEPARstart{P}{roduct} codes \cite{Elias:errorfreecoding54} have gained considerable attention due to their suitability for low-complexity iterative decoding \cite{Tanner,Pyndiah98}. Usually, product codes are constructed as two or three dimensional arrays, where each dimension is encoded by a short algebraic code. This choice allows the use of low-complexity \ac{SISO} \cite{Pyndiah98} or algebraic (e.g., bounded distance) \cite{Abramson,KoetterPC,Haeger2018} decoders for the component codes. \ac{RM} codes \cite{RM1} and their majority logic decoding \cite{reed} were introduced roughly one year after product codes.
	Since then, \ac{RM} codes have been analyzed intensively both from the code structure \cite{KKMPSU17} and decoding \cite{Dumer04,Dumer:List06} points of view. Interest in \ac{RM} codes has recently grown due to their close relationship with polar codes \cite{arikan2009channel,stolte2002rekursive}. It has been shown \cite{KKMPSU17} that they achieve capacity under \ac{MAP} decoding over \ac{BEC}.
	
	Product codes based on \ac{RM} component codes have been considered, e.g., in \cite{Pyndiah98,Chiaraluce} for the case where the component codes are extended Hamming and \ac{SPC} codes. This choice of component codes has been  considered in wireless communication systems (see, e.g., \cite{berrou2005overview,valenti2009channel,IEEE80216}) thanks to the availability of low-complexity \ac{SISO} decoders for \ac{SPC} and extended Hamming codes.\footnote{Extended Hamming codes can be efficiently decoded, for instance, by exploiting their compact trellis representation or by employing the sub-optimum Chase-Pyndiah decoder proposed in \cite{Pyndiah98}.}
	
	In this letter, product codes with \ac{SPC} and/or extended Hamming component codes are considered. The emphasis is on the short and moderate blocklength regimes due to the increasing interest in efficient error correcting codes for short packet transmissions required by emerging applications (see, e.g., \cite{Liva2019:Survey}). Typically, encoding of product codes makes use of the component codes' systematic encoders, while decoding is performed iteratively. An equivalent code can be obtained by using non-systematic encoders for the component codes. In particular, an \ac{RM} product code construction which directly maps the code structure onto the iterated Kronecker product of the $2\times 2$ Hadamard kernel is considered \cite{SA05}. Similar observations were used to increase the throughput of polar codes by reducing the decoding latency/complexity \cite{PA13,BCL19}. In \cite{PA13}, the authors proposed a construction that allows interpreting a polar code as a $2$-dimensional product code. The construction allows using \ac{SC} decoders row- and column-wise to reduce the complexity/latency with respect to the case where \ac{SC} decoding is performed over the larger polar code. Similarly, \cite{BCL19} proposed designing of product codes, where the component codes are polar codes. The focus was on reducing the latency by proposing a two-stage decoding, where the \ac{SC} decoder of the large polar code is used only if the iterative product code decoder does not converge to a valid codeword. 
	In this paper, we make use of this observation to employ standard polar code \ac{SC} and \ac{SCL} decoding algorithms, instead of \ac{BP} decoding, for a class product codes of large practical interest.\footnote{Recently, the \ac{SC} decoding of product codes with \ac{SPC} component codes has been introduced in \cite{coskun17} by using \ac{SPC} kernels.} Simulations on the \ac{B-AWGN} channel show that \ac{SCL}  decoding with small list sizes performs as good as \ac{BP} decoding  for the considered cases.	
	The benefits of \ac{SCL} decoding of product codes extends beyond the potential channel coding gains with respect to \ac{BP} decoding. In fact, \ac{SCL} decoding allows low-complexity decoding of the concatenation of the product code with a high-rate outer  code, as proposed for polar codes in \cite{tal15}. The concatenation provides remarkable gains over the product code alone, and over a \ac{BP} decoder which jointly decodes the outer code and the inner product code. This behaviour is characterized by a weight enumerator analysis of the concatenated product codes, restricted to the minimum distance terms. Another important advantage of list decoding is to reduce the number of pilots for channel estimation when communicating over fading channels with unknown channel state \cite{Coskun:SCC19}.
	
	The work is organized as follows. In Section \ref{sec:prelim}, we provide preliminaries needed for the rest of the work. In Section \ref{sec:spcpc_as_pc}, we revisit the connection between \ac{RM} codes and product codes with \ac{RM} component codes. The concatenation with a high-rate outer code is discussed in Section \ref{sec:concatenation}. Numerical results are provided in Section \ref{sec:numerical_results}. Conclusions follow in Section \ref{sec:conc}.
	
	\section{Preliminaries}\label{sec:prelim}
	
	\subsection{Product Codes}\label{sec:prodcodes}
	
	A $\mu$-dimensional  $(n,k,d)$ product code $\code$  \cite{Elias:errorfreecoding54} is obtained by \emph{iterating} $\mu$ binary linear block codes $\code_1, \code_2, \ldots, \code_\mu$. Let $\code_\ell$ be the $\ell$th component code with parameters $(n_\ell,k_\ell,d_\ell)$, where $n_\ell$, $k_\ell$, and $ d_\ell$ are its blocklength, dimension, and minimum Hamming distance, respectively and $\G_{\ell}$ is the generator matrix. Then, the parameters of the resulting product code are the multiplication of the individual ones. Similarly, its generator matrix is given as
	\begin{equation}
	\label{eq:gen_prod}
	\G = \G_{1} \otimes \G_{2} \otimes \ldots \otimes \G_{\mu}.
	\end{equation}
	
	Although characterizating the distance spectrum of a product code is an elusive problem (with a few exceptions, see \cite{caire}), the minimum distance $d$ and the multiplicity $A_d$ of codewords with weight $d$ can be obtained as
	$d = \prod_{\ell=1}^{\mu} d_\ell$ and  $A_d=\prod_{\ell=1}^\mu A^{(\ell)}_{d_\ell}$,
	where $A^{(\ell)}_w$ is the multiplicity of the codewords having weight of $w$ in the $\ell$th component code.
	Examples of minimum distances and minimum weight multiplicities for some $2$-dimensional product codes based on \ac{SPC} and extended Hamming codes are provided in Table \ref{tab:min_dist}. While achieving relatively large minimum distances, product codes based on these component codes suffer, in general, from large minimum weight multiplicity \cite{Chiaraluce}. This  observation is important to understand the gains attainable by concatenating product codes with high-rate outer codes.
	
	\renewcommand{\arraystretch}{1.15}
	\begin{table}
		\caption{Minimum distances and multiplicities of some 
			product codes (eH = extended Hamming code).\\[-3mm]}
		\begin{center}
			\begin{tabular}{cccc}
				\hline\hline
				$(n,k,d)$ & $\mathcal{C}_1$ & $\mathcal{C}_2$ & $A_d$ \\
				\hline
				$(128,105,4)$ & SPC $(16,15)$ & SPC $(8,7)$ & $3360$  \\
				$(128,77,8)$ & eH $(16,11)$ & SPC $(8,7)$ & $3920$  \\
				$(256,121,16)$ & eH $(16,11)$ & eH $(16,11)$ & $19600$  \\
				%		$(256,165,8)$ & eH $(16,11)$ & SPC $(16,15)$ & $16800$  \\
				$(256,225,4)$ & SPC $(16,15)$ & SPC $(16,15)$ & $14400$ \\
				$(1024,693)$ & SPC $(64,63)$ & eH $(16,11)$ & $282240$  \\
				%$(1024,806)$ & eH $(32,26)$ & SPC $(32,31)$ & $8$ & $615040$ \\
				%$(1024,961)$ & SPC $(32,31)$ & SPC $(32,31)$ & $4$ & $246016$ \\
				\hline\hline
			\end{tabular}
		\end{center}
		\label{tab:min_dist}
	\end{table}
	\renewcommand{\arraystretch}{1}
	
	\subsection{Reed-Muller Codes}
	
	\label{sec:RM}
	The construction of an $r$-th order \ac{RM} code of length $n = 2^m$ and dimension $k = 1+{m \choose 1}+{m \choose 2}+\dots+{m \choose r}$, denoted by $\mathrm{RM}(r,m)$ with $0\leq r\leq m$, starts by defining the $n \times n$ Hadamard matrix $\G_{n}=\GK_2^{\otimes m}$ where $\GK_{2}^{\otimes m}$ denotes the $m$-fold Kronecker product of the Hadamard kernel
	\[
	\GK_{2}\triangleq\begin{bmatrix}
	1 & 0 \\ 
	1 & 1
	\end{bmatrix}.
	\]
	The \ac{RM} code generator matrix $\G$ is obtained by removing the rows of $\G_{n}$ with  weight lower than $2^{m-r}$.
	We denote the set containing the indices of discarded rows (frozen bits) as $\mathcal A$ and let $\vecv = (v_1, v_2,\ldots,v_k)$ contain the indices in the complement set $\mathcal A^C$ in an ascending order, i.e., $v_1<v_2<\ldots<v_k$. We define the entry at the $i$th row and $j$th column of the $k\times n$ matrix $\F$ as
	\begin{equation}
	F_{i,j} = \left\{\begin{array}{lll}
	1 &\text{if}\,\,\, j = v_i \\[1mm]
	0 & \text{otherwise.}
	\end{array}
	\right.
	\end{equation}
	It follows that
	\begin{equation}
	\label{eq:gen_RM}
	\G = \F\G_{n}.
	\end{equation}
	The $n$-dimensional frozen bit vector $\vecf$, where $f_i = 0$ for all $i\in \mathcal A$ and $f_i = 1$ otherwise, is also obtained by summing all the rows of $\F$, i.e.,
	\begin{equation}
	\label{eq:frozen_f}
	\vecf = \boldsymbol{e}\F
	\end{equation}
	where $\bm{e}$ is the length-$k$ all-one vector.
	For the encoding, an $n$-bit vector $\vecu=(u_1,u_2,\dots,u_n)$ is defined, where $u_i=0$ for all $i\in \mathcal A$ and the remaining $k$ elements of $\vecu$ consists of information bits. Encoding proceeds as $\vecc=\vecu\G_{n}$.
	
	\section{Reed-Muller Product Codes}\label{sec:spcpc_as_pc}
	
	Consider a product code where the $\ell$th component code $\code_\ell$ is an $\mathrm{RM}(r_\ell,m_\ell)$ code with the corresponding row-selecting matrix $\F_\ell$ (frozen bit vector $\vecf_\ell$). It is a subcode of the $\mathrm{RM}(r_1+r_2+\ldots+r_\mu,m_1+m_2+\ldots+m_\mu)$ code \cite{SA05}. 
	For the $\ell$th component code, the generator matrix is obtained via \eqref{eq:gen_RM} as
	\begin{equation}
	\label{eq:gen_RM_component}
	\G_\ell = \F_\ell\GK_2^{\otimes m_\ell}.
	\end{equation}
	Using the mixed-product property of the Kronecker product, the frozen bit vector for the resulting product code is obtained.
	\begin{prop}\label{prop_one} The generator matrix of the product code obtained by iterating $\mu$ \ac{RM} codes $\mathrm{RM}(r_1,m_1)$, $\mathrm{RM}(r_2,m_2)$, $\ldots$, $\mathrm{RM}(r_\mu,m_\mu)$ is given by
		\begin{equation}\label{eq:fK}
		\G = \F\GK_2^{\otimes (m_1+m_2+\ldots+m_\mu)}
		\end{equation}
		where $\F = \F_1\otimes\F_2\otimes\ldots\otimes\F_\mu$, resulting in
		\begin{equation}
		\label{eq:frozen_set}
		\vecf = \vecf_1\otimes\vecf_2\otimes\ldots\otimes\vecf_\mu.
		\end{equation}
	\end{prop}
	\begin{proof}\renewcommand{\qedsymbol}{}
		Follows from the application of the mixed-product property to the combination of \eqref{eq:gen_prod}, \eqref{eq:frozen_f} and \eqref{eq:gen_RM_component}:
		%\vspace{-2mm}
		\begin{align*}
		(&\F_{1}\GK_2^{\otimes m_1})\otimes(\F_{2}\GK_2^{\otimes m_2})\otimes\dots\otimes(\F_\mu\GK_2^{\otimes m_\mu}) = \\ &(\F_{1}\otimes\F_2\otimes\dots\otimes\F_\mu)(\GK_2^{\otimes m_{1}}\otimes\GK_2^{\otimes m_{2}}\otimes\dots\otimes\GK_2^{\otimes m_\mu}). ~\hfill\ensuremath{\square}
		&
		\end{align*}
	\end{proof}
	\vspace{-7mm}
	\begin{example}
		Consider a two-dimensional product code with a $(2,1)$ repetition code and a $(4,3)$ \ac{SPC} code as component codes with $\F_{1}=\begin{bmatrix}
		0 & 1
		\end{bmatrix}$ yielding $\vecf_{1}=\begin{bmatrix}
		0 & 1
		\end{bmatrix}$ and 
		\[
		\F_{2}=\begin{bmatrix}
		0 & 1 & 0 & 0 \\ 
		0 & 0 & 1 & 0 \\
		0 & 0 & 0 & 1
		\end{bmatrix}
		\]
		yielding $\vecf_{2}=\begin{bmatrix}
		0 & 1 & 1 & 1
		\end{bmatrix}$. Then, the product code generator matrix is obtained via \eqref{eq:fK} where
		\[
		\F=\begin{bmatrix}
		0 & 0 & 0 & 0 & 0 & 1 & 0 & 0 \\ 
		0 & 0 & 0 & 0 & 0 & 0 & 1 & 0 \\
		0 & 0 & 0 & 0 & 0 & 0 & 0 & 1
		\end{bmatrix}
		\]
		with $\vecf=\begin{bmatrix}
		0 & 0 & 0 & 0 & 0 & 1 & 1 & 1
		\end{bmatrix}$ via \eqref{eq:frozen_set}.
	\end{example}
	Proposition $1$ provides an interpretation of product codes with \ac{RM} component codes as \ac{RM} subcodes, where the frozen bit positions are given as \eqref{eq:frozen_set}. This enables the use of the \ac{SC} and \ac{SCL} decoding algorithms derived for \ac{RM} and polar codes \cite{stolte2002rekursive,arikan2009channel,tal15} to decode this class of product codes.
	
	\begin{remark}
		Equivalent codes, defined by different frozen bit vectors, can be obtained as the Kronecker product is not commutative. The definition of the different frozen bits vectors is related to the order with which the component codes generator matrices are iterated in \eqref{eq:gen_prod}. 
	\end{remark}
	
	\subsection{Successive Cancellation (List) Decoding}
	
	As for polar codes \cite{arikan2009channel}, \ac{SC} decoding for the construction illustrated above estimates bit $u_i$, $i = 1,2,\dots,n$, by using the channel observation $\vecy=(y_1,y_2,\ldots,y_n)$ and the previous decisions $\hat{u}_1,\hat{u}_2,\ldots,\hat{u}_{i-1}$, taking into account the code constraints imposed by the set $\mathcal A$, i.e., $\hat{u}_i$ is set to zero if $i\in\mathcal A$. The decoding equations are defined by $\GK_{2}$ and used recursively to compute the soft-information for the bits. In \ac{SCL} decoding \cite{tal15}, two hypothesis are kept open for each bit $u_i$ if it is not a frozen bit. Whenever the number of the hypothesis exceeds a given maximum list size $L$, they are pruned by keeping the most likely ones according to the computed metrics. At the final stage, the decoder outputs the most likely candidate as the message estimate.
	
	\begin{remark}
		For polar codes, it was shown in \cite{tal15} that a performance close to the one of an \ac{ML} decoder can be attained with a sufficiently large list size. This was demonstrated by computing a numerical lower bound on the \ac{ML} decoding error probability via Monte Carlo simulation, where the correct codeword is introduced artificially in the final list, prior to the final selection. If for a specific list size $L$ the simulated error probability is close to the numerical \ac{ML} decoding lower bound, then increasing the list size $L$ would not yield visible improvement. We shall see in Section \ref{sec:numerical_results} that the same principle applies to \ac{SCL} decoding of product codes.
	\end{remark}
	
	\section{Concatenation with a High-Rate Outer Code}
	\label{sec:concatenation}
	Following \cite{tal15}, we analyze the performance under \ac{SCL} decoding of product codes based on \ac{SPC} and extended Hamming component codes, in concatenation with a high-rate outer code. The outer code is used to test each codeword in the final list when \ac{SCL} decoding is used. Among the survivors, the most likely one is chosen as the final decision. A reason to analyze such concatenation lies (besides in the obviously expected performance improvement) in the fact that actual schemes employing product codes may make use of an error detection code to protect the product code information message. We may hence consider sacrificing (part of) the error detection capability for a larger coding gain. Following the construction adopted in the IEEE 802.16 standard \cite{IEEE80216}, we consider product codes with systematic encoding.
	
	For product codes, large gains are expected by adding a high-rate outer code, especially at moderate-low error rates. This follows from the fact that product codes with \ac{RM} component codes are characterized by a fairly large multiplicity of minimum weight codewords, as already observed in Section \ref{sec:prodcodes}. When using a product code based on \ac{RM} component codes to transmit over a memory-less binary-input output-symmetric channel with Bhattacharyya parameter $\beta$, the block error probability under \ac{BP} decoding can be well approximated by the \ac{ML} decoding \ac{TUB}
	\begin{equation}
	P_B \simeq A_d \beta^d \label{eq:bat}
	\end{equation} 
	already at moderate error rates \cite{Chiaraluce}. Recalling Remark 2, the error probability under \ac{SCL} decoding is thus limited by the \ac{ML} decoding performance already at moderate error rates.
	By a suitable choice of the outer code, the multiplicity of minimum weight codewords may be considerably lowered, hence, reducing its contribution to the overall block error probability under \ac{ML} decoding. This may potentially yield remarkable gains also under sub-optimum \ac{BP}/\ac{SCL} decoding.
	We analyze the impact of the outer code in a concatenated ensemble setting from a weight distribution perspective by focusing on the minimum weight terms only.
	
	\subsection{Average Weight Distribution of Concatenated Ensembles}
	
	We consider the concatenation of an $(n_{\mathsf i},k_{\mathsf i})$ inner product code $\code_{\mathsf{i}}$ with an $(n_{\mathsf o},k_{\mathsf o})$ high-rate outer code $\code_{\mathsf{o}}$. Note that $k_{\mathsf i}=n_{\mathsf o}$. We denote by $d$ the minimum distance of the inner product code. We further define the generator matrices of $\code_{\mathsf{i}}$ and $\code_{\mathsf{o}}$ as $\G_{\mathsf{i}}$ and $\G_{\mathsf{o}}$, respectively.
	\begin{definition}[Concatenated Ensemble] The (serially) concateneted ensemble $\ensemble\left(\code_{\mathsf{o}},\code_{\mathsf{i}}\right)$ is the set  of all codes with generator matrix of the form
		$
		\G=\G_{\mathsf{o}}\bm{\Pi}\G_{\mathsf{i}}
		$,
		where $\bm\Pi$ is an $n_{\mathsf o}\times n_{\mathsf o}$ permutation matrix.
	\end{definition}
	Denote the outer code weight enumerator by $A^{\mathsf o}_j$. The minimum-weight input-output weight enumerator of the inner product code is given by $A^{\mathsf i}_{j,d}$.
	The expected number of weight-$d$ codewords for a code drawn randomly from $\ensemble\left(\code_{\mathsf{o}},\code_{\mathsf{i}}\right)$ is
	\begin{equation}
	\bar{A}_d=\sum_{j=1}^{n_{\mathsf o}} \frac{A^{\mathsf o}_j A^{\mathsf i}_{j,d}}{{n_{\mathsf o} \choose j}}.\label{eq:WEF}
	\end{equation}
	The expected multiplicity of weight-$d$ codewords $\bar{A}_d$ can be used in \eqref{eq:bat} to obtain an estimate of the ensemble average error probability in the low error probability regime.
	If $\G_{\mathsf{i}}$ is in systematic form, $A^{\mathsf i}_{j,d}$ is easily computed from the input-output weight enumerators of the component codes \cite[Thm.~1]{El-Khamy05}.
	
	\begin{example}Consider the $(128,77)$ systematic product code with $(16,11)$ extended Hamming and $(8,7)$ \ac{SPC} component codes, which has minimum distance $8$ with a multiplicity of $3920$. The code is concatenated with an outer \ac{CRC}$-7$ code with generator polynomial $g(x)=x^7 + x^3 + 1$. The resulting code is a member of a concatenated ensemble with an expected number of weight-$8$ codewords given by $\bar{A}_8\simeq 26.4$, i.e., the multiplicity of weight-$8$ codewords is reduced, on average, by two orders of magnitude. The contribution of these codewords to the ensemble average error probability is reduced significantly. Hence, the \ac{TUB} shall be approached only at low error rates.
	\end{example}
	
	Note that the generator matrix of the product code constructed according to Proposition \ref{prop_one} is not in systematic form. Assume the inner code generator matrix $\G_{\mathsf{i}}$ to be systematic. The overall code generator matrix can be written as 
	\begin{align*}
	\G=\G_{\mathsf{o}}\bm{\Pi}\G_{\mathsf{i}}
	=\G_{\mathsf{o}}\bm{\Pi}\bm{S}\Gnsys
	=\G_{\mathsf{mo}}\Gnsys
	\end{align*}	
	where $\bm{\Pi}$ is the interleaver matrix, $\bm{S}$ is a $k_{\mathsf i} \times k_{\mathsf i}$ non-singular matrix and $\Gnsys=\bm{S}^{-1}\G_{\mathsf{i}}$ is the non-systematic generator matrix according to Proposition \ref{prop_one}. Furthermore, $\G_{\mathsf{mo}}$ is defined to be the product $\G_{\mathsf{o}}\bm{\Pi}\bm{S}$. Thus, the SCL decoding can be used for the inner product code, where the modified outer code with generator matrix $\G_{\mathsf{mo}}$ is used to test the codewords of the final list prior to a decision.
	
	\section{Numerical Results}
	\label{sec:numerical_results}
	
	We provide simulation results for two product codes, over the \ac{B-AWGN} channel. The results are provided in terms of \ac{CER} vs. \ac{SNR}, where the \ac{SNR} is expressed as $E_b/N_0$ ratio ($E_b$ is here the energy per information bit, and $N_0$ the single-sided noise power spectral density.) For both codes, the \ac{SCL} decoding performance is compared to Gallager's \ac{RCB}~\cite{gallager65-01a} and the \ac{RCUB} from \cite[Thm.~16]{Polyanskiy10:BOUNDS}. As a reference, the performance under \ac{BP} is provided with a maximum number of iterations set to $100$. The component codes are decoded by \ac{MAP} \ac{SISO} decoding over the component code trellis.  For both product codes, the concatenation with a high-rate outer code is also considered. For the short construction,  \ac{BP} decoding of the concatenated scheme is also provided, where the product code Tanner graph is modified by adding a check node representing the outer code constraints (as for the component codes, the outer code is decoded within the node by a \ac{MAP} \ac{SISO} decoder). The \ac{TUB} in the tighter form of \cite[Eq. 3]{Chiaraluce} is also provided.
	
	The first product code is the $(128,77)$ code from Example 1, whose performance is depicted in Figure 1a. In particular, the component codes $\code_1$ and $\code_2$ are $(16,11)$ extended Hamming and $(8,7)$ \ac{SPC} codes, respectively. List decoding with $L=4$ is sufficient to approach the performance of \ac{BP} decoding. With $L=8$, the \ac{SCL} decoder tightly matches the \ac{ML} lower bound below \ac{CER} of $10^{-2}$. The gap to the \ac{RCUB}  is limited to $2$ dB at \ac{CER} of $10^{-6}$. In the same figure, the performance of a $(128,77)$ polar code under \ac{SCL} decoding is provided. While the \ac{ML} decoding performance of the polar code and of the product code are very close, the polar code requires a smaller list size to approach the \ac{ML} lower bound. Figure 1b shows the performance by concatenating the $(128,77)$ product code with an outer \ac{CRC} code with generator polynomial $g(x)=x^7 + x^3 + 1$, leading to a $(128,70)$ code. The performance of the concatenated scheme is provided for two interleavers between the inner and outer code. The label ``no interleaver'' denotes the trivial interleaver, i.e., $\bm{\Pi}$ is chosen to be the $ k_{\mathsf i} \times k_{\mathsf i}$ identity matrix, while in the second case a random interleaver is used. The concatenation with the trivial interleaver performs remarkably well under \ac{SCL} decoding. At  a \ac{CER} of $10^{-6}$, \ac{SCL} decoding of the concatenated code achieves gains up to $1.4$ dB over the original product code. The gains attained by \ac{SCL} decoding over \ac{BP} decoding range from $1$ dB at a \ac{CER} of $10^{-2}$ to $1.4$ dB at a \ac{CER} $\approx10^{-5}$. The gap to the \ac{RCUB} is $0.5$ dB at a \ac{CER} $\approx10^{-7}$. In this specific case, the omission of an interleaving stage yields a code performing better than the ensemble average. For sake of completeness, the performance of a concatenation employing a randomly generated interleaver is provided. The result tightly approaches, in this case, the expected ensemble performance approximated by the TUB.
	
	\begin{figure}    
		\centering
		\begin{subfigure}{0.85\columnwidth}
			\includegraphics[width=\columnwidth]{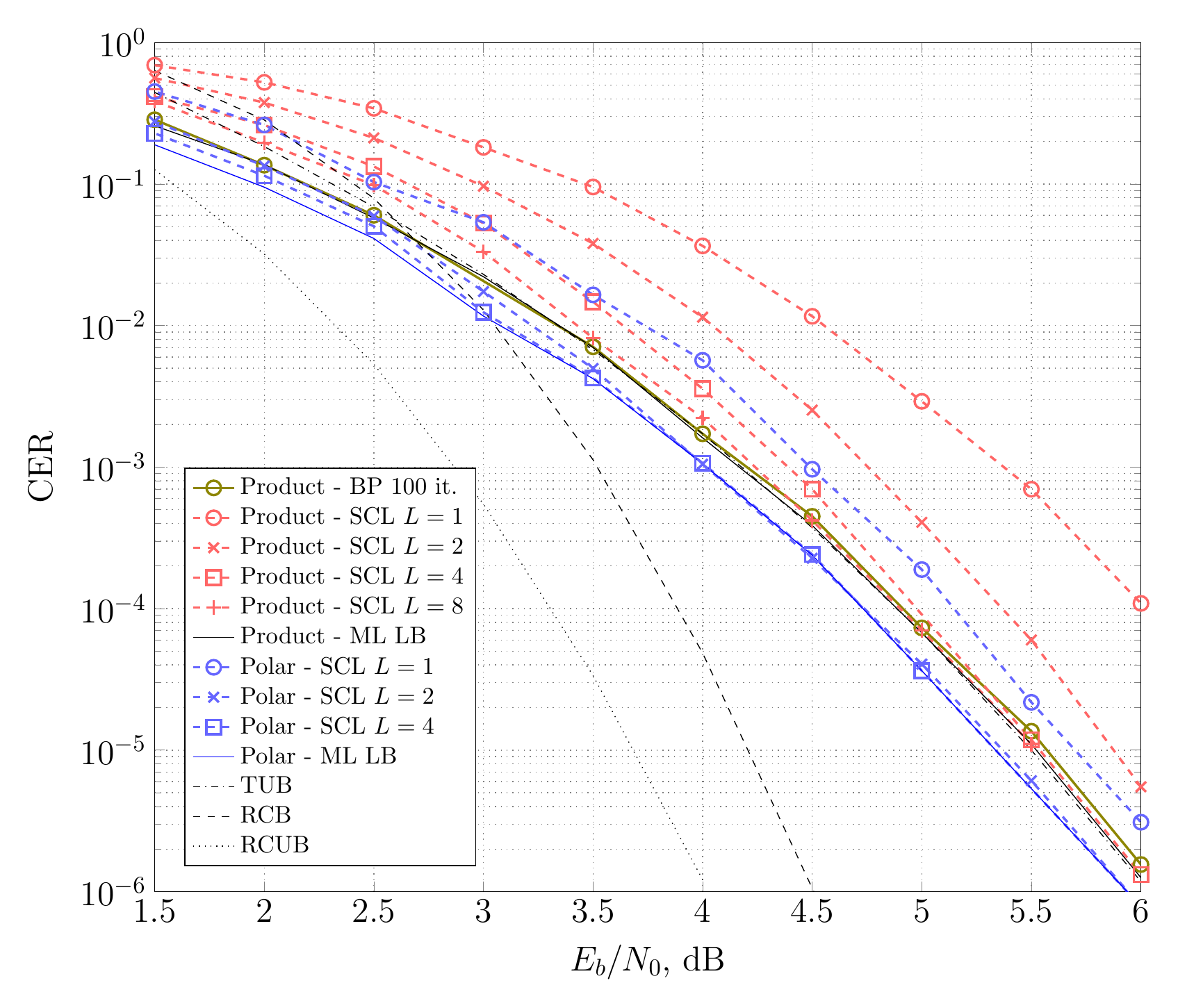}\\[-5.5mm]
			\caption{$(128,77)$ product code}
			\label{fig:scl_128_77_RM}
		\end{subfigure}
		\begin{subfigure}{0.85\columnwidth}
			\includegraphics[width=\columnwidth]{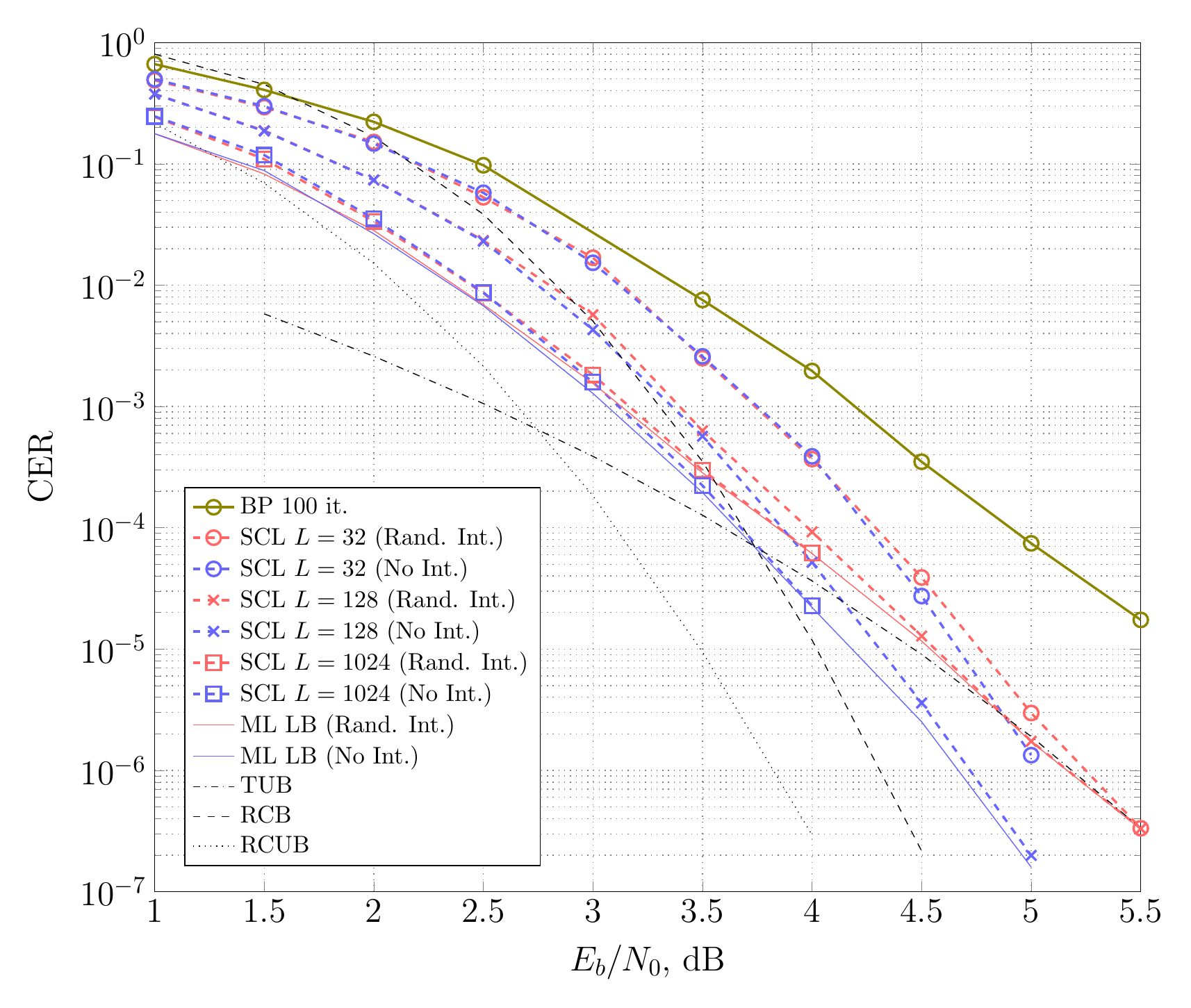}\\[-5.5mm]
			\caption{$(128,70)$ concatenated code}
			\label{fig:scl_128_70_RM}
		\end{subfigure}
		\caption{\ac{CER} vs. \ac{SNR}  under \ac{SCL} decoding with various list sizes for the (a) $(128,77)$ product and polar codes, and (b) $(128,70)$ concatenated product codes, compared with \ac{BP} decoding.}\label{fig:combined}
	\end{figure}
	
	A longer $(1024,693)$ product code has been constructed by choosing $\code_1$ and $\code_2$ to be the $(16,11)$ extended Hamming and the $(64,63)$ \ac{SPC} codes, respectively. The results of the product codes with\footnote{For this case, the performance under \ac{BP} is not provided. The reason is that the addition of the outer code check node in the product code Tanner graph resulted in large performance degradation, due to the emergence of a number of small trapping sets for the \ac{BP} decoder.} and without outer codes are provided in Figure 2. The outer \ac{CRC} code has a generator polynomial $g(x)=x^{10} + x^9 + x^5 + x^4 + x + 1$, leading to a $(1024,683)$ code. The concatenated schemes needs a larger list than the one required by the product code alone to approach the \ac{ML} lower bound, especially at high error rates. In particular,  for long blocklengths the required list size increases due to the sub-optimal choice of a large number of non-frozen bits enforced by the specific product code construction.
	
	\begin{figure}    
		\centering
		\begin{subfigure}{0.85\columnwidth}
			\includegraphics[width=\columnwidth]{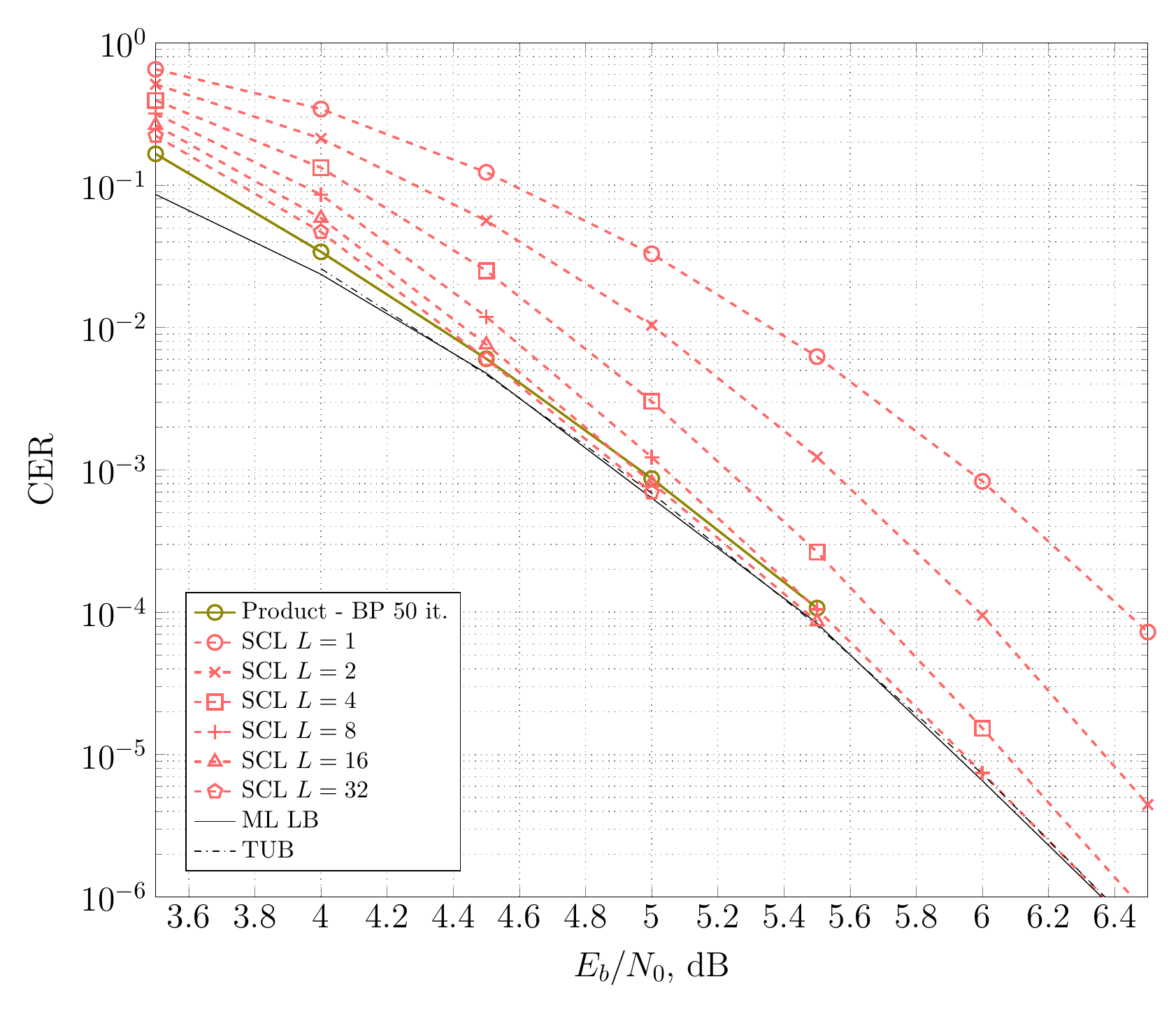}\\[-5.5mm]
			\caption{$(1024,693)$ product code}
			\label{fig:scl_1024_693_RM}
		\end{subfigure}
		\begin{subfigure}{0.85\columnwidth}
			\includegraphics[width=\columnwidth]{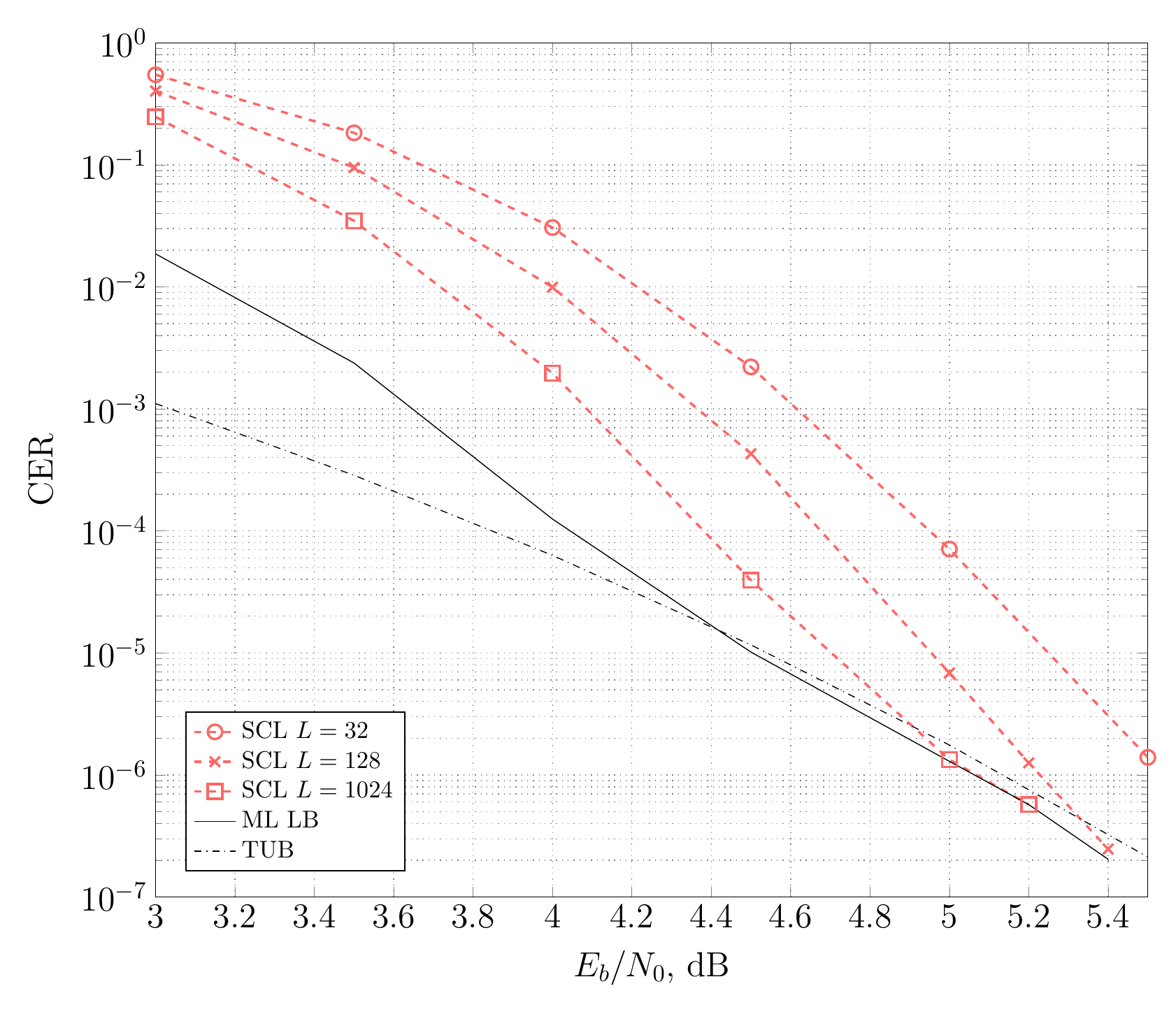}\\[-5.5mm]
			\caption{$(1024,683)$ concatenated code}
			\label{fig:scl_1024_683_RM}
		\end{subfigure}
		\caption{\ac{CER} vs. \ac{SNR}  under \ac{SCL} decoding with various list sizes for (a) $(1024,693)$ product code and (b) $(1024,683)$ concatenated one.}\label{fig:combined2}
	\end{figure}

	\section{Conclusions}
	\label{sec:conc}
	
	Successive cancellation list (SCL) decoding of product codes with single parity-check and extended Hamming component codes has been investigated. SCL decoding relies on a product code description based on the $2\times 2$ Hadamard kernel, which enables interpreting the code as a Reed-Muller subcode. With small list sizes, \ac{SCL} decoding performs as good as (and sometimes it outperforms) belief propagation decoding. 
	Larger gains are attained by concatenating an inner product code with an outer high-rate code. For specific concatenated constructions, a performance within a few tenths of a decibel from finite length bounds at \ac{CER}$\approx10^{-7}$ is achieved.
	
	\section*{Acknowledgement}
	The authors would like to thank Gerhard Kramer, Peihong Yuan, the Associate Editor and Anonymous Reviewers  for their insightful comments that helped to significantly improve the paper.

\end{document}